\RequirePackage{fix-cm}
\documentclass{svjour3}                      
\smartqed  
\usepackage{graphicx}
\usepackage{amssymb}
\usepackage{comment}
\usepackage{enumitem}
\usepackage{lipsum}
\usepackage{ntheorem}
\usepackage{url}
\newtheorem{Property}{Property}

\usepackage{ntheorem}
\theorembodyfont{\upshape}
\newtheorem{Comment}{Comment}
\AtBeginDocument{%
  \setlength{\oddsidemargin}{\dimexpr(\paperwidth-\textwidth)/2-1in}%
  \setlength{\evensidemargin}{\oddsidemargin}%
  \setlength{\topmargin}{%
    \dimexpr(\paperheight-\textheight)/2-\headheight-\headsep-1in}%
}

\usepackage{amsmath}
\DeclareMathOperator*{\argmin}{arg\,min}

\begin{document}

\title{On the Complexity of the Geometric Median Problem with Outliers}

\titlerunning{Geometric Median with Outliers}        

\author{Vladimir Shenmaier\thanks{%
Supported by the program of fundamental scientific researches of the SB RAS, project 0314-2019-0014.
}}

\authorrunning{V.V. Shenmaier} 

\institute{%
V.V. Shenmaier \at
Sobolev Institute of Mathematics, 4 Koptyug avenue, 630090 Novosibirsk, Russia\\
\email{shenmaier@mail.ru}           
}

\date{}

\def\makeheadbox{}
\maketitle

\begin{abstract}
In the Geometric Median problem with outliers, we are given a finite set of points in $d$-dimensional real space and an integer $m$, the goal is to locate a new point in space (center) and choose $m$ of the input points to minimize the sum of the Euclidean distances from the center to the chosen points.
This problem can be solved ``almost exactly'' in polynomial time if $d$ is fixed and admits an approximation scheme PTAS in high dimensions.
However, the complexity of the problem was an open question.
We prove that, if the dimension of space is not fixed, Geometric Median with outliers is strongly NP-hard, does not admit approximation schemes FPTAS unless P$=$NP, and is W[$1$]-hard with respect to the parameter~$m$.
The proof is done by a reduction from the Independent Set problem.
Based on a similar reduction, we also get the NP-hardness of closely related geometric $2$-clustering problems in which it is required to partition a given set of points into two balanced clusters minimizing the cost of median clustering.
Finally, we study Geometric Median with outliers in $\ell_\infty$ space and prove the same complexity results as for the Euclidean problem.
\end{abstract}

\section{Introduction}
We give an answer to the open question on the complexity of the Geometric Median problem with outliers, which is an extension of the classical problem of finding the geometric median (``minisum point'') of given $n$ points in Euclidean space.

Let $\|\cdot\|$ denote values of the Euclidean norm.
The problem we consider is formulated as follows:\medskip

\noindent\textbf{Geometric Median with outliers.}
Given an $n$-element set $X$ in space $\mathbb R^d$ and an integer~$m$.
Find a center $c\in\mathbb R^d$ and an $m$-element subset $S\subseteq X$ to minimize the value of
$$cost(S,c)=\sum_{x\in S}\|x-c\|.$$

\begin{figure}
\centering
\includegraphics[scale=1]{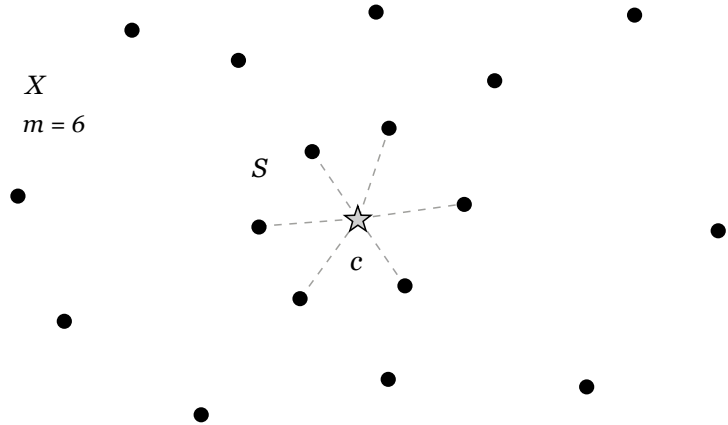}
\caption{An instance of Geometric Median with outliers and its solution}
\label{fig:Example}
\end{figure}

This problem may also be referred to as {\it Capacitated Euclidean $1$-Median\/}.
In fact, we need to find a center $c\in\mathbb R^d$ minimizing the total distance from $c$ to $m$ nearest input points (see Fig.~\ref{fig:Example}).
In an equivalent version, we need to find a center $c\in\mathbb R^d$ maximizing the cardinality of a subset $S\subseteq X$ for which the value of $cost(S,c)$ does not exceed a given upper bound.

As in the usual Geometric Median problem, where $m=n$, the input points are considered as clients or demand points, while the desired center, as a location for placing a facility to serve the clients.
The $n-m$ clients which are removed from this service in the solution are called {\it outliers\/}.
The problem with outliers arises naturally in the following situations:
\begin{itemize}[label={$\bullet$}]
\item The facility to be placed at the center we are looking for has a limited capacity and may not serve all the demand points.\smallskip
\item There exist an upper limit for the transportation cost, so we need to remove a minimum possible number of clients from the service to satisfy this limit.\smallskip
\item The data contains noise and errors.
In this case, a few most distant clients may exert a disproportionately strong influence over the final solution and correspond to the least robust input points.\smallskip
\item The discovered outliers do not fit the rest of the data and are worthy of further investigation.
In particular, once identified, they can be used to discover anomalies in the data.
\end{itemize}

Besides the practical considerations mentioned above, the problem is theoretically interesting.
Strictly speaking, no polynomial-time algorithms are known even for the problem without outliers, where we find the geometric median of the whole set~$X$.
However, one can say that the usual Geometric Median problem is polynomially solvable ``almost exactly'' since, e.g., the algorithm from~\cite{Cohen} computes its $(1+\varepsilon)$-approximate solution in time $\mathcal{O}\big(nd\log^3\frac{n}{\varepsilon}\big)$.

The problem with outliers seems to be cardinally harder due to the exponential number of the subsets $S\subseteq X$.
The first idea how to solve it is using the property that optimal subsets consist of $m$ input points nearest to some point in space.
It allows to find an optimal solution by enumerating the cells of the $m$-order Voronoi diagram for the set $X$ and calculating the geometric medians of the subsets corresponding to these cells.
So, by using the known algorithms for constructing high-order Voronoi diagrams \cite{ERS,ESS} and for finding geometric medians \cite{Cohen}, we compute a $(1+\varepsilon)$-approximate solution of the Geometric Median problem with outliers in time $\mathcal{O}\big(n^{d+1}md\log^3\frac{m}{\varepsilon}\big)$.
Thus, in the case of fixed $d$, we can solve the problem ``almost exactly'' in polynomial time.

If the dimension of space is not fixed, a more effective idea is to use the approximation algorithm from \cite{Shen2019,Shen2020} for the following general clustering problem, which contains Geometric Median with outliers:\medskip

\noindent\textbf{General Problem.}
Given points $x_1,\dots,x_n$ in space $\mathbb R^d$, integers $k,m\ge 1$, unit distance costs $f_{ij}\ge 0$, and powers $\alpha_{ij}\in[0,\alpha]$, $i=1,\dots,k$, $j=1,\dots,n$, where $\alpha$ is some constant.
Find disjoint subsets $S_1,\dots,S_k\subseteq\{1,\dots,n\}$ of total cardinality $m$ and select a tuple $c_1,\dots,c_k\in\mathbb R^d$ to minimize the value of
$$\sum_{i=1}^k\sum_{j\in S_i}f_{ij}\|x_j-c_i\|^{\alpha_{ij}}.$$

\noindent The framework suggested in \cite{Shen2019,Shen2020,Shen2021d} allows to find a $(1+\varepsilon)^\alpha$-approximate solution of General Problem in time $\mathcal{O}\big(n^{k\lceil\log(2/\varepsilon)/\varepsilon\rceil+1}kd\,\big)$ for any fixed $\varepsilon\in(0,1]$.
This framework is based on constructing a polynomial-cardinality set of points which approximate all the points of space with respect to the distances to the input points.
Thus, in our case, we have an approximation scheme PTAS with running time $\mathcal{O}\big(n^{\lceil\log(2/\varepsilon)/\varepsilon\rceil+1}d\,\big)$.\medskip

\noindent\textbf{Related problems.}
ElGindy and Keil~\cite{ElGindy} consider the equivalent version of the Geometric Median problem with outliers where it is required to find a maximum-cardinality subset of input points satisfying a given upper bound for the cost of the geometric median.
They suggest an $\mathcal{O}(n^{2.5}\log^4n)$-time exact algorithm for the two-dimensional case with $\ell_1$ or $\ell_\infty$ distances.

Two well-known single location problems closest to Geometric Median with outliers are Smallest $m$-Enclosing Ball and $m$-Variance.
The first consists of finding $m$ input points minimizing the radius of the ball enclosing these points.
In the second, we need to find $m$ input points minimizing the sum of the squared distances from these points to their mean.
In high dimensions, both problems are strongly NP-hard \mbox{\cite{Shen2013,Shen2015,KP}} but admit approximation schemes with running time $\mathcal{O}\big(n^{\lceil 1/\varepsilon\rceil}d\,\big)$ \cite{AHV2005} and $\mathcal{O}\big(n^{\lceil 2/\varepsilon\rceil+1}d\,\big)$ \cite{Shen2012} respectively.

Finally, note that the discrete version of $k$-Median with outliers, in which all the centers must be selected from a given finite set, and also $k$-Means with outliers, Uncapacitated Facility Location with outliers, and other similar problems were considered by Charikar et al. \cite{Charikar}, Chen \cite{Chen2008}, Cohen-Addad et al. \cite{CFS}, and Krishnaswamy et al. \cite{KLS}.
In particular, they suggest constant-factor algorithms for these problems in the case of non-fixed~$k$.\medskip

\noindent\textbf{Our contributions.}
Surprisingly, the complexity of Geometric Median with outliers was an open question.
We prove that, in high dimensions, this problem is strongly NP-hard, does not admit approximation schemes FPTAS unless P$=$NP, and is W[$1$]-hard with respect to the parameter~$m$.
The proof is based on a reduction from the Independent Set problem.
The main idea of this reduction is constructing instances in which the geometric median of any subset of points corresponding to vertices of a given graph is well approximated by the mean of these points, while the total distance from the points to their mean monotonously depends on the number of edges between the corresponding vertices.

By using a similar reduction from the problem of finding a maximum bisection in a $3$-regular graph, we get the NP-hardness and the non-existence of approximation schemes FPTAS for the following closely related problems:\medskip

\noindent\textbf{Equal-Size Geometric $2$-Median.}
Given an $n$-element set $X$ in space $\mathbb R^d$, where $n$ is even.
Find centers $x,y\in\mathbb R^d$ and a partition of $X$ into two equal-size subsets $S$ and $X\setminus S$ to minimize the value of $cost(S,x)+cost(X\setminus S,\,y)$.\medskip

\noindent\textbf{Balanced Geometric $2$-Median.}
Given an $n$-element set $X$ in space $\mathbb R^d$.
Find centers $x,y\in\mathbb R^d$ and a partition of $X$ into two subsets $S$ and $X\setminus S$ to minimize the value of $|S|\cdot cost(S,x)+|X\setminus S|\cdot cost(X\setminus S,\,y)$.\medskip

Additionally, we study the {\it $\ell_\infty$-Median problem with outliers\/}, which consists of finding a center minimizing the sum of the $\ell_\infty$ distances to $m$ input points.
For this problem, the same complexity results are proved as for the Euclidean case.

\section{Complexity of the Euclidean problem}
In this section, we prove the hardness of finding an optimum $m$-element subset in the case when the distances between input points and the desired center are defined by the Euclidean norm.

\subsection{Reduction from Independent Set} 
We construct a reduction to Geometric Median with outliers from the classic NP-hard problem of determining the existence of an independent set of a given cardinality in a general graph.

Let $G$ be any undirected graph on some $n$-element set of vertices~$V$.
First, construct the following auxiliary graph~$G'$.
The vertices of $G'$ are the same as in $G$ and the set of edges contains all the edges of $G$ and also $n-deg_v$ copies of the loop edge $(v,v)$ for each vertex $v\in V$, where $deg_v$ is the degree of $v$ in~$G$. 
So each vertex in $G'$ is incident to exactly $n$ edges (see Fig.~\ref{fig:F1}).
Then fix an arbitrary orientation on the edges of $G'$: for every non-loop edge, choose an endpoint of this edge which it is ``outgoing from'' and one which it is ``incoming to''.
For definiteness, we will assume that every loop edge is ``incoming to'' its endpoint.
Denote by $E$ the set of edges of~$G'$.
Next, map each vertex $v\in V$ to the point $x_v\in\mathbb R^{E\cup V}$ with the following coordinates: $x_v(v)=M$, where $M$ is some big integer which will be specified later; $x_v(e)=1$ for every edge $e\in E$ outgoing from $v$ and $x_v(e)=-1$ for ones incoming to $v$; all the other coordinates are zero (see Fig.~\ref{fig:F1}).
Finally, define the instance of the Geometric Median problem with outliers corresponding to the graph $G$ as the set $X=\{x_v\mid v\in V\}$.

\begin{figure}
\centering
\includegraphics[scale=1]{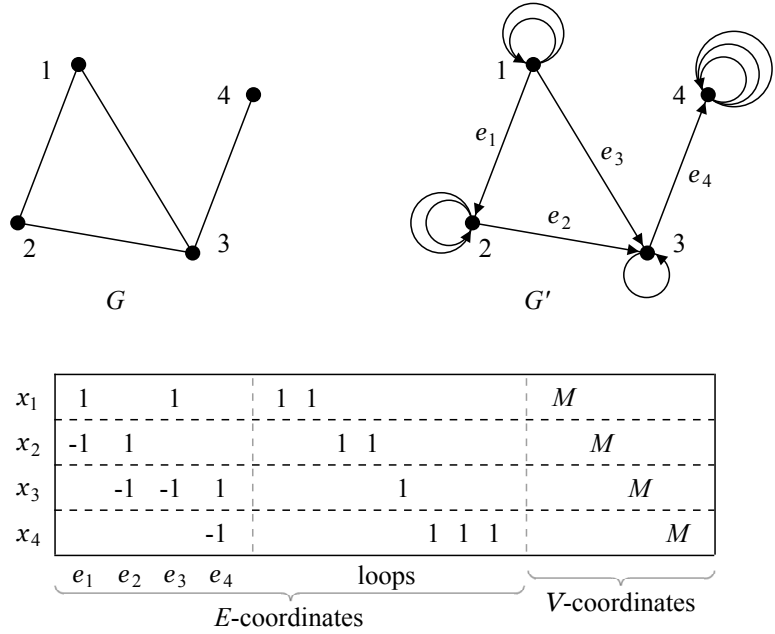}
\caption{The original graph $G$, the auxiliary graph $G'$, and the vectors $x_v$, $v\in V$}
\label{fig:F1}
\end{figure}

\begin{Comment}
Adding the loop edges and selecting an edge orientation in the auxiliary graph $G'$ provide that all the vectors $x_v$, $v\in V$, have the same number of the $\pm 1$ coordinates and, for any two vectors $x_v,x_u$, the sets of their non-zero coordinates are either disjoint or intersecting by the only coordinate $e\in E$, for which $x_v(e)=-x_u(e)$.
As a corollary, the distance between the vectors $x_v$ and $x_u$ for any adjacent different vertices $u,v\in V$ is greater than that for non-adjacent ones.
Indeed, $\|x_v-x_u\|^2=2M^2+\|x_v^E-x_u^E\|^2$, where $x^E$ denotes the projection of $x$ into space $\mathbb R^E$, while the value of $\|x_v^E-x_u^E\|^2$ is $2(n-1)+(1+1)^2=2n+2$ if the vertices $u$ and $v$ are adjacent and is $2n$ otherwise.
\end{Comment}

\begin{Comment}
Setting the coordinates $x_v(v)$, $v\in V$, equal to a big value of $M$ provides that the geometric median of any subset $Y\subseteq X$ is close to its mean, while, as it will be proved, the total distance from the points $x_v\in Y$ to their mean has an almost affine dependence on the sum of pairwise squared distances between the corresponding vectors~$x_v^E$.
\end{Comment}

\noindent\textbf{Idea of the reduction.}
We will prove that, given an $m$-element subset of vertices $S\subseteq V$, the minimum value of $cost(X_S,c)$, where $X_S=\{x_v\mid v\in S\}$, over all $c\in\mathbb R^d$ monotonously depends on the number $\ell_S$ of the edges in the original graph $G$ connecting the vertices of $S$ to each other.
Therefore, if the set $X_S$ is an optimal solution of Geometric Median with outliers for the given $m$, then $S$ is an $m$-element subset of vertices with the minimum number~$\ell_S$.\medskip

For any finite set $Y\subset\mathbb R^{E\cup V}$, denote by $c(Y)$ and $\mu(Y)$ the mean and the geometric median of this set respectively:
$$c(Y)=\frac{1}{|Y|}\sum_{y\in Y}y\quad\mbox{ and }\quad\mu(Y)=\argmin_{c\in\mathbb R^{E\cup V}}cost(Y,c).$$

Let $S\subseteq V$ be any subset of vertices with an arbitrary cardinality $m\ge 3$.
Put $y=\mu(X_S)-c(X_S)$, $\delta=\|y\|$, $z_v=x_v-c(X_S)$, and $\displaystyle\zeta_v=\sum_{e\in E}z_v^2(e)$ for each $v\in S$.
Then $\displaystyle\sum_{v\in S}z_v$ is the zero vector and the following property holds:

\begin{Property} 
$(a)$ $\displaystyle\zeta_v\in\Big[n-\frac{2n}{m}+\frac{2n}{m^2},\,n+\frac{n}{m}-\frac{2n}{m^2}\Big]\subset(0.55n,\,1.13n)$;\\
$(b)$ $\|z_v\|=\sqrt{A^2+\zeta_v}$, where $A=M\sqrt{1-1/m}$.
\end{Property}

\begin{proof}
(a) Given a vector $x\in\mathbb R^{E\cup V}$, let $E(x)=\{e\in E\mid x(e)\ne 0\}$.
Then it is easy to see that the set $E(c(X_S))$ consists of the loops in $S$ and the edges connecting $S$ and $V\setminus S$.
Hence, we have $|E(c(X_S))|=nm-2\ell_S$.
Next, the set $E(z_v)$ consists of all the elements of $E(c(X_S))$ and also the edges connecting $v$ and the other vertices of~$S$.
So $|E(z_v)|=nm-2\ell_S+\Delta^v_S$, where $\Delta^v_S$ is the degree of $v$ in the subgraph of $G$ induced by~$S$.
For $\Delta^v_S$ coordinates $e\in E(z_v)$, the values of $z_v(e)$ are $\pm 1$;
for $n-\Delta^v_S$ coordinates, these values are $\pm(1-1/m)$;
for the other coordinates from $E(z_v)$, these values are $\pm 1/m$.
Then
\begin{eqnarray*}
\zeta_v=\Delta^v_S+(n-\Delta^v_S)\Big(1-\frac{1}{m}\Big)^2+\frac{n(m-1)-2\ell_S+\Delta^v_S}{m^2}=\\
n\Big(1-\frac{1}{m}\Big)^2+\frac{2\Delta^v_S}{m}-\frac{\Delta^v_S}{m^2}+\frac{n}{m}-\frac{n}{m^2}-\frac{2\ell_S}{m^2}+\frac{\Delta^v_S}{m^2}=
n-\frac{n}{m}-\frac{2\ell_S}{m^2}+\frac{2\Delta^v_S}{m}.
\end{eqnarray*}
But $\ell_S\ge\Delta^v_S$ and $\Delta^v_S<n$,
so
$\displaystyle\zeta_v\le n-\frac{n}{m}-\frac{2\Delta^v_S}{m^2}+\frac{2\Delta^v_S}{m}<n+\frac{n}{m}-\frac{2n}{m^2}$.
The latter is maximized when $m=4$, therefore, we have $\zeta_v<9n/8<1.13n$.

On the other hand, we have $\ell_S=\Delta^v_S+\ell_{S\setminus\{v\}}$, $\displaystyle\ell_{S\setminus\{v\}}\le\frac{(m-1)(m-2)}{2}$, and $\Delta^v_S\ge 0$, so
\begin{eqnarray*}
\zeta_v=n-\frac{n}{m}-\frac{2\ell_S}{m^2}+\frac{2\Delta^v_S}{m}=n-\frac{n}{m}-\frac{2\ell_{S\setminus\{v\}}}{m^2}-\frac{2\Delta^v_S}{m^2}+\frac{2\Delta^v_S}{m}\ge\\
n-\frac{n}{m}-\frac{(m-1)(m-2)}{m^2}=n-1-\frac{n-3}{m}-\frac{2}{m^2}.
\end{eqnarray*}
But it can be easily proved that $\displaystyle n-1-\frac{n-3}{m}-\frac{2}{m^2}\ge n-\frac{2n}{m}+\frac{2n}{m^2}$ for all $n\ge m\ge 3$.
It follows that $\zeta_v\ge 5n/9>0.55n$.\medskip

\noindent(b) Obviously, we have $\displaystyle\|z_v\|^2=M^2\Big(1-\frac{1}{m}\Big)^2+(m-1)\Big(\frac{M}{m}\Big)^2+\zeta_v=A^2+\zeta_v$, which implies the required equation.
\hfill$\Box$
\end{proof}

The main geometric statement underlying the proposed reduction is that the distance between $c(X_S)$ and $\mu(X_S)$ is very close to zero for big $M$: $\displaystyle\delta<\frac{4.5n}{mA}$ if $A\ge 100nm$.
This statement will be proved in Section~2.2 (Lemma~\ref{medmean}).

\begin{lemma}\label{medcut}
Suppose that $A\ge 100nm$.
Then, for some $\gamma\in[-1,1]$, we have
$$cost(X_S,\mu(X_S))=mA+\frac{n(m-1)}{2A}+\frac{\ell_S}{mA}+\gamma\frac{2.3n^2m}{A^3}.$$
\end{lemma}

\begin{proof}
By the cosine theorem, we have $cost(X_S,\mu(X_S))=$
\begin{eqnarray*}
\sum_{v\in S}\sqrt{\|z_v\|^2+\|y\|^2-2\langle y,z_v\rangle}=\sum_{v\in S}\sqrt{A^2+\zeta_v+\delta^2-2\langle y,z_v\rangle},
\end{eqnarray*}
where $\langle\cdot\,,\cdot\rangle$ is the dot product.
Property~1, Lemma~\ref{medmean}, and the condition for~$A$ yield that
$$\zeta_v+\delta^2-2\langle y,z_v\rangle\le 1.13n+\frac{(4.5n)^2}{(mA)^2}+2\frac{4.5n}{mA}\sqrt{A^2+1.13n}<4.2n,$$
so $\displaystyle\Big|\frac{\zeta_v+\delta^2-2\langle y,z_v\rangle}{A^2}\Big|<\frac{4.2n}{A^2}<0.001$.
On the other hand, Taylor's theorem (in the Lagrange remainder form) implies the equation
\begin{eqnarray*}
\sqrt{1+\varepsilon}=1+\frac{\varepsilon}{2}-\theta\frac{\varepsilon^2}{7.99}\mbox{ for some }\theta\in[0,1]\mbox{ if }|\varepsilon|\le 0.001.
\end{eqnarray*}
Therefore,
\begin{eqnarray*}
cost(X_S,\mu(X_S))=
\sum_{v\in S}A\Big(1+\frac{\zeta_v+\delta^2-2\langle y,z_v\rangle}{2A^2}-\theta_v\frac{(\zeta_v+\delta^2-2\langle y,z_v\rangle)^2}{7.99A^4}\Big),
\end{eqnarray*}
where $\theta_v\in[0,1]$.
Since $\displaystyle\sum_{v\in S}z_v$ is the zero vector, the sum of the terms $\langle y,z_v\rangle$ is zero.
Then, taking into account the inequality $\displaystyle\delta<\frac{4.5n}{mA}$ and the above observations, we obtain that $cost(X_S,\mu(X_S))=$
\begin{eqnarray*}
\sum_{v\in S}\Big(A+\frac{\zeta_v}{2A}\Big)+\theta_1\frac{(4.5n)^2}{2mA^3}-\theta_2\frac{(4.2n)^2m}{7.99A^3}=
mA+\sum_{v\in S}\frac{\zeta_v}{2A}+\gamma\frac{2.3n^2m}{A^3},
\end{eqnarray*}
where $\theta_1,\theta_2\in[0,1]$ and $\gamma\in[-1,1]$.

Next, calculate the value of $\displaystyle\sum_{v\in S}\zeta_v$.
Given a vertex $v\in S$, consider the projection $x^E_v\in\mathbb R^E$ of the vector $x_v$ into space $\mathbb R^E$.
It is easy to see that the value of $\zeta_v$ equals the squared distance from the vector $x^E_v$ to the mean of these vectors.
But it is well-known (e.g., see \cite{Inaba,Shen2012}) that the sum of such squared distances equals the sum of all the pairwise squared distances divided by $2m$:
$$\sum_{v\in S}\zeta_v=\frac{1}{2m}\sum_{v\in S}\sum_{u\in S}\|x^E_v-x^E_u\|^2.$$
At the same time, by the construction of the vectors $x^E_v$, each pairwise squared distance $\|x^E_v-x^E_u\|^2$ equals $2n+2$ if the vertices $v$ and $u$ are adjacent and equals $2n$ otherwise (see Comment~1).
So we have
$$\sum_{v\in S}\zeta_v=\frac{2n(m^2-m)+4\ell_S}{2m}=n(m-1)+\frac{2\ell_S}{m}.$$
It follows the required equation.
The lemma is proved.
\hfill$\Box$
\end{proof}

\begin{theorem}\label{th1}
Geometric Median with outliers is strongly NP-hard, does not admit approximation schemes FPTAS unless P$=$NP, and is W$[1]$-hard with respect to the parameter~$m$.
\end{theorem}

\begin{proof}
Suppose that $G$ is any $n$-vertex graph, $S$ is any $m$-element subset of its vertices, where $m\ge 3$, and put $M=123nm$.
Then $A\ge M\sqrt{2/3}>100nm$ and, by Lemma~\ref{medcut}, we have the equation $cost(X_S,\mu(X_S))=f(n,m,\ell_S,\gamma)$, where
$$f(n,m,\ell,\gamma)=mA+\frac{n(m-1)}{2A}+\frac{\ell}{mA}+\gamma\frac{2.3n^2m}{A^3}$$
and $\gamma\in[-1,1]$.
By the choice of $M$, the absolute value of the latest term in this expression is at most $0.00023$ times of $\displaystyle\frac{1}{mA}$.
Therefore, $S$ is an independent set in the graph $G$ if and only if $cost(X_S,\mu(X_S))<f(n,m,1,-1)$.
Thus, the Independent Set problem is reduced to Geometric Median with outliers.
Taking into account that $M$ is an integer bounded by a polynomial in the length of the input, it gives the strong NP-hardness of the latter problem.

Moreover, by the above, if $\ell\ge 1$, then
$$f(n,m,\ell,-1)-f(n,m,0,1)\ge\frac{1-2\cdot 0.00023}{mA}>\frac{0.999}{123n^3}>\frac{0.008}{n^3}.$$
On the other hand, both values $f(n,m,\ell,-1)$ and $f(n,m,0,1)$ are less than $mM\le 123n^3$, which implies that Geometric Median with outliers is NP-hard to approximate within a factor of $\displaystyle 1+\frac{0.008}{123n^6}$.
But, for an arbitrary polynomial $p(n)$, any approximation scheme FPTAS allows to get a $\displaystyle\Big(1+\frac{1}{p(n)}\Big)$-approximation in polynomial time.
Therefore, the existence of such schemes is impossible unless P$=$NP.

Finally, note that the constructed reduction can be represented as a pa\-ra\-me\-te\-rized reduction of the W[$1$]-hard Independent Set problem parameterized by $m$ \cite{Downey} to Geometric Median with outliers parameterized by~$m$.
It follows the statement on the W[$1$]-hardness.
The theorem is proved.
\hfill$\Box$
\end{proof}

\subsection*{Balanced Geometric $2$-Median clustering}
Based on the same construction of the set $X$, we can give similar reductions from the well-known NP-hard problem of finding a maximum bisection in a $3$-regular graph~\cite{Feige} to the Equal-Size and Balanced Geometric $2$-Median problems.

Indeed, in a $3$-regular graph, the number of the edges connecting an $m$-element subset of vertices $S$ with its complement equals $3m-2\ell_S$.
Therefore, by Lemma~\ref{medcut}, the minimum value of $cost(X_S,\mu(X_S))+cost(X_{V\setminus S},\mu(X_{V\setminus S}))$ in the case when $m=n-m$ corresponds to a bisection with the maximum number of crossing edges.
It gives the NP-hardness and the non-existence of approximation schemes FPTAS for Equal-Size Geometric $2$-Median.

To get the similar result for Balanced Geometric $2$-Median, we need to note that the value of the objective function of this problem on the bipartition $X_S$, $X\setminus X_S$ equals the expression
$$f(n,m,\ell_1,\gamma_1,\ell_2,\gamma_2)=mf(n,m,\ell_1,\gamma_1)+(n-m)f(n,n-m,\ell_2,\gamma_2),$$
where $\ell_{1,2}\le 3n/2$ and $\gamma_{1,2}\in[-1,1]$.
But this expression can be approximated by the value of $\displaystyle m^2M\sqrt{1-\frac{1}{m}}+(n-m)^2M\sqrt{1-\frac{1}{n-m}}$ for big~$M$.
It follows that the minimum of $f(n,m,\ell_1,\gamma_1,\ell_2,\gamma_2)$ is attained when $m=\lfloor n/2\rfloor$ and, by the above observations, corresponds to a maximum bisection.

\subsection{Location of the geometric median}
In this section, we prove the upper bound for the distance between the geometric median and the mean of the set $X_S$ underlying Lemma~\ref{medcut}.

Let $S$ be an arbitrary $m$-element subset of~$V$.
Since the mean of any single- and two-element set is also its geometric median, we will assume that $m\ge 3$.

\subsection*{Rough estimations}
First, we prove some rough estimation for the value $\|\mu(X_S)-c(X_S)\|$ which will be needed for the proof of the stronger upper bound.

We start by introducing some notations.
Given a set $Y\subset\mathbb R^{E\cup V}$, denote by ${\it Aff}(Y)$ the affine hull of this set.
Denote by $\Gamma$ the subspace in space $\mathbb R^{E\cup V}$ consisting of all the vectors $y$ with the zero edge coordinates: $y(e)=0$, $e\in E$.
Given a vertex $v$ of the graph $G$, define the point $y_v\in\mathbb R^{E\cup V}$ for which \mbox{$y_v(v)=M$} and all the other coordinates are zero.
In other words, the vector $y_v$ is the orthogonal projection of $x_v$ into the space~$\Gamma$.
Note that the distance from the mean $c(Y_S)$ of the set $Y_S=\{y_v\mid v\in S\}$ to each of its elements is $A=M\sqrt{1-1/m}$.

\begin{lemma}\label{m1}
The geometric median of the set $Y_S$ coincides with its mean.
Moreover, if a vector $z\in\mathbb R^{E\cup V}$ is orthogonal to the affine space ${\it Aff}(Y_S)$, then the minimum of the value $cost(Y_S,z+x)$ subject to $x\in{\it Aff}(Y_S)$ is attained at $x=c(Y_S)$.
\end{lemma}

\begin{proof}
The derivative of the function $cost(Y_S,x)$ along any coordinate $u$ is
\begin{eqnarray*}
\Big(\sum_{v\in S}\|x-y_v\|\Big)'_u=\sum_{v\in S}\frac{(x-y_v)(u)}{\|x-y_v\|}.
\end{eqnarray*}
So the value of this derivative at the point $x=c(Y_S)$ is
\begin{eqnarray*}
\frac{1}{A}\Big(m\,c(Y_S)-\sum_{v\in S}y_v\Big)(u)=0.
\end{eqnarray*}
It follows that the point $c(Y_S)$ is the geometric median of the set $Y_S$.

Then the derivative of the function $f(x)=cost(Y_S,x)$ along any direction $\phi$ in the affine space ${\it Aff}(Y_S)$ is equal to zero at the point $x=c(Y_S)$: $f'_\phi(c(Y_S))=0$.
On the other hand, the derivative of the function $g(x)=cost(Y_S,z+x)$ along any direction $\phi$ in ${\it Aff}(Y_S)$ is
\begin{eqnarray*}
g'_\phi(x)=\Big(\sum_{v\in S}\sqrt{\|z\|^2+\|x-y_v\|^2}\Big)'_\phi=
\sum_{v\in S}\frac{\|x-y_v\|\,\big(\|x-y_v\|\big)'_\phi}{\sqrt{\|z\|^2+\|x-y_v\|^2}}.
\end{eqnarray*}
So the value of this derivative at the point $x=c(Y_S)$ is
\begin{eqnarray*}
\frac{A}{\sqrt{\|z\|^2+A^2}}\sum_{v\in S}\big(\|x-y_v\|\big)'_\phi(c(Y_S))=
\frac{A}{\sqrt{\|z\|^2+A^2}}f'_\phi(c(Y_S))=0.
\end{eqnarray*}
It follows that the point $x=c(Y_S)$ is the minimum of the function $cost(Y_S,z+x)$ subject to $x\in{\it Aff}(Y_S)$.
The lemma is proved.
\hfill$\Box$
\end{proof}

\begin{lemma}\label{m2}
Suppose that $A\ge 300n$ and $\displaystyle dir_v=\frac{y_v-c(Y_S)}{\|y_v-c(Y_S)\|}$, $v\in S$.
Then, for each $v\in S$, the following holds: $\big|\big\langle\mu(X_S)-c(X_S),\,dir_v\big\rangle\big|<1.45\sqrt n$.
\end{lemma}

\begin{proof}
Let $\mu'$ be the orthogonal projection of the point $\mu(X_S)$ into the space~$\Gamma$.
Then $\mu'\in{\it Aff}(Y_S)$ and $cost(Y_S,\mu')\le cost(X_S,\mu(X_S))$ since the set $Y_S$ is the projection of the set $X_S$ into ${\it Aff}(Y_S)$.
Next, consider the affine subspace $H_v$ in ${\it Aff}(Y_S)$ orthogonal to the vector $y_v-c(Y_S)$ and passing through the point~$\mu'$.
Denote by $\mu_v$ the intersection point of this subspace with the line pas\-sing through the points $c(Y_S)$ and~$y_v$ (see Fig.~\ref{fig:F2}).
Then $\mu_v$ is the point nearest to $y_v$ in~$H_v$.
On the other hand, it is easy to see that $H_v$ can be obtained by shifting the affine subspace ${\it Aff}(Y_S\setminus\{y_v\})$ by the vector $z=\mu_v-c(Y_S\setminus\{y_v\})$ and $z$ is orthogonal to ${\it Aff}(Y_S\setminus\{y_v\})$.
Hence, applying Lemma~\ref{m1} to the set $Y_S\setminus\{y_v\}$, we get that the point $\mu_v$ is the minimum of the function $cost(Y_S\setminus\{y_v\},\,x)$ subject to $x\in H_v$.
It follows that
$cost(Y_S,\mu_v)\le cost(Y_S,\mu')\le cost(X_S,\mu(X_S))\le cost(X_S,c(Y_S))$.

\begin{figure}
\centering
\includegraphics[scale=1]{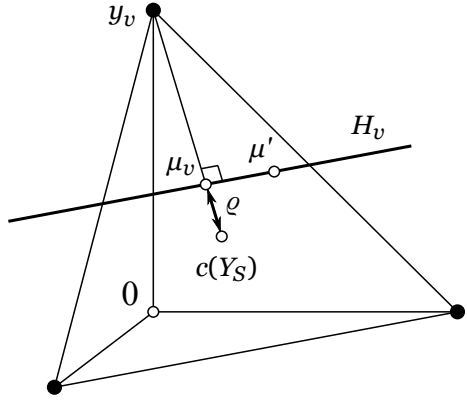}
\caption{Proof of Lemma~\ref{m2}}
\label{fig:F2}
\end{figure}

By the cosine theorem, we have
$$cost(Y_S,\mu_v)=A-\varrho+(m-1)\sqrt{A^2+\varrho^2+\frac{2A\varrho}{m-1}},$$
where $|\varrho|=\|\mu_v-c(Y_S)\|$.
At the same time,
$$cost(X_S,c(Y_S))=m\sqrt{A^2+n}\le m\Big(A+\frac{n}{2A}\Big).$$
Then, since $cost(Y_S,\mu_v)\le cost(X_S,c(Y_S))$, we get the inequality
\begin{eqnarray*}
\sqrt{A^2+\varrho^2+\frac{2A\varrho}{m-1}}\le
\frac{m\big(A+\frac{n}{2A}\big)+\varrho-A}{m-1}=
A+\frac{\frac{nm}{2A}+\varrho}{m-1},
\end{eqnarray*}
which implies that
\begin{eqnarray*}
A^2+\varrho^2+\frac{2A\varrho}{m-1}\le
A^2+2A\frac{\frac{nm}{2A}+\varrho}{m-1}+ \frac{\big(\frac{nm}{2A}+\varrho\big)^2}{(m-1)^2}.
\end{eqnarray*}
Taking into account that $m\ge 3$ and, therefore, $\displaystyle\frac{m}{m-1}\le 1.5$, we have
\begin{eqnarray*}
\varrho^2\le 1.5n+\frac{9n^2}{16A^2}+\frac{3n\varrho}{4A}+\frac{\varrho^2}{4}.
\end{eqnarray*}
So $|\varrho|<1.45\sqrt n$ if $A\ge 300n$.

It remains to note that the vectors $\mu(X_S)-\mu'$ and $c(X_S)-c(Y_S)$ are orthogonal to the space ${\it Aff}(Y_S)$ and, therefore, to the vector~$dir_v$.
But the vector $\mu'-\mu_v$ is also orthogonal to $dir_v$ by the choice of~$\mu_v$.
Thus, we obtain that
$$\big\langle\mu(X_S)-c(X_S),\,dir_v\big\rangle=\big\langle\mu_v-c(Y_S),\,dir_v\big\rangle,$$
which equals $|\varrho|$ in absolute value.
The lemma is proved.
\hfill$\Box$
\end{proof}

\begin{lemma}\label{m3}
Suppose that $A\ge 300n$.
Then $\|\mu(X_S)-c(X_S)\|<1.8\sqrt{nm}$.
\end{lemma}

\begin{proof}
For every $v\in S$, consider the plane $P_v$ containing the points $y_v$, $c(Y_S)$, and the origin~$\textbf 0$.
Note that this plane is orthogonal to the affine space ${\it Aff}(Y_S)$ since the line $c(Y_S)\textbf 0$ is orthogonal to ${\it Aff}(Y_S)$ and is contained in~$P_v$.
At the same time, the projection $\mu'$ of the point $\mu(X_S)$ into the space $\Gamma$ lies in ${\it Aff}(Y_S)$.
So the orthogonal projection $\mu''$ of the point $\mu'$ into the plane $P_v$ belongs to the intersection of $P_v$ with ${\it Aff}(Y_S)$, i.e., lies on the line $y_vc(Y_S)$ (see Fig.~\ref{fig:F3}).
\begin{figure}
\centering
\includegraphics[scale=1]{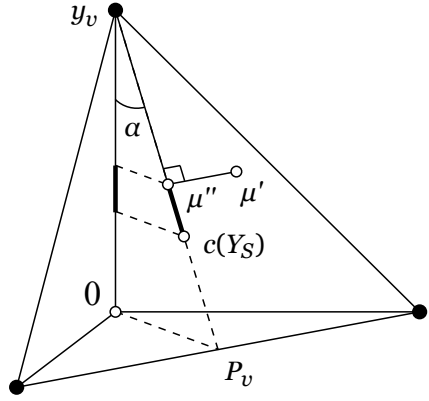}
\caption{Proof of Lemma~\ref{m3}}
\label{fig:F3}
\end{figure}
Then
\begin{eqnarray*}
\big(\mu(X_S)-c(X_S)\big)(v)=\big(\mu'-c(Y_S)\big)(v)=\big(\mu''-c(Y_S)\big)(v)=\\
\|\mu''-c(Y_S)\|\cos\alpha=\big|\big\langle\mu(X_S)-c(X_S),\,dir_v\big\rangle\big|\cos\alpha,
\end{eqnarray*}
where $\alpha$ is the angle between the lines $y_v\textbf 0$ and $y_vc(Y_S)$.
Taking into account the equation $\cos\alpha=\sqrt{1-1/m}$ and Lemma~\ref{m2}, it follows that the absolute value of the $v$-th coordinate of the vector $\mu(X_S)-c(X_S)$ is less than $1.45\sqrt n\sqrt{1-1/m}$.
On the other hand, for each edge $e\in E$ incident to $S$, the absolute value of the $e$-th coordinate of this vector is at most $1$ since the geometric median is a convex combination of the points $x_v$, $v\in S$.
For every other coordinate $i\in E\cup V$ and every $v\in S$, we have $x_v(i)=0$, which implies that $\big(\mu(X_S)-c(X_S)\big)(i)=0$.
Thus, we obtain the inequality $\|\mu(X_S)-c(X_S)\|^2<1.45^2nm+nm<3.2nm$, so $\|\mu(X_S)-c(X_S)\|<1.8\sqrt{nm}$.
The lemma is proved.
\hfill$\Box$
\end{proof}

\subsection*{A tighter estimation}
Now, we get the tighter estimation for the geometric median location which is used in Section~2.1.
We start by proving the following geometric properties of the vectors $z_v=x_v-c(X_S)$, $v\in S$.

\begin{Property} 
If $M\ge 300n$, then the angle between the vectors $z_u$ and $z_v$ for different $u,v\in S$ lies in the interval $\displaystyle\Big(\arccos\frac{-0.99\ }{m-1},\,\arccos\frac{-1.01\ }{m-1}\Big)$.
\end{Property}

\begin{proof}
Indeed, $\displaystyle\langle z_u,z_v\rangle=\sum_{a\in V}z_u(a)z_v(a)+\sum_{e\in E}z_u(e)z_v(e)$.
The first term in this expression is $2M(1-1/m)(-M/m)+(m-2)M^2/m^2=-M^2/m$;
the second is at least $-1-(2n-2)(1-1/m)/m>-1-2n/m$ and is at most $(nm-2n)/m^2$.
Then, taking into account Pro\-per\-ty~1(a), we obtain that the cosine of the angle between the vectors $z_u$ and $z_v$ is between the values $\displaystyle a=\frac{-M^2/m-1-2n/m}{A^2+0.55n}$ and $b=\displaystyle\frac{-M^2/m+(nm-2n)/m^2}{A^2+1.13n}$.
It remains to note that $A^2=M^2(1-1/m)$ and that $M$ is sufficiently big, so we have $\displaystyle a>\frac{-1.01\ }{m-1}$ and $\displaystyle b<\frac{-0.99\ }{m-1}$.
\hfill$\Box$
\end{proof}

\begin{Property} 
If $M\ge 300n$, then
$\displaystyle\sum_{v\in S}\langle y,z_v\rangle^2\le\frac{m+1.515}{2}\|y\|^2\max_{v\in S}\|z_v\|^2$.
\end{Property}

\begin{proof}
By Property~2, the angles between the vectors $z_u$ and $\pm z_v$ for $u\ne v$ are at least $\displaystyle\alpha=\pi-\arccos\frac{-1.01\ }{m-1}=\arccos\frac{1.01}{m-1}$.
So there exists at most one vector $z_u$, $u\in S$, such that the angles between $y$ and $\pm z_u$ are at most some $\beta<\alpha/2$.
Then $\displaystyle\sum_{v\in S}\langle y,z_v\rangle^2\le\|y\|^2\max_{v\in S}\|z_v\|^2\,t$, where $\displaystyle t=\cos^2\beta+(m-1)\cos^2(\alpha-\beta)$.
Note that $\displaystyle t=\frac{m}{2}+\frac{\cos(2\beta)+(m-1)\cos(2\alpha-2\beta)}{2}$ and that this expression has the only local extremum on $[0,\alpha/2]$, at the point $\beta=0$.
Therefore, we have $\displaystyle t=t(\beta)\le\max\{t(0),t(\alpha/2)\}=t(\alpha/2)=m\cos^2(\alpha/2)=\frac{m}{2}(1+\cos\alpha)$, which equals $\displaystyle\frac{m}{2}\Big(1+\frac{1.01}{m-1}\Big)=\frac{m}{2}+\frac{1.01}{2}+\frac{1.01}{2(m-1)}\le\frac{m+1.515}{2}$ if $m\ge 3$.
\hfill$\Box$
\end{proof}

%

\begin{Property} 
If $M\ge 300n$, then
$\displaystyle\Big|\sum_{v\in S}\langle y,z_v\rangle\zeta_v\Big|<1.061n\,\|y\|\max_{v\in S}\|z_v\|$.
\end{Property}

\begin{proof}
By Property~1(a), we have $\displaystyle\zeta_v=n-\frac{0.5n}{m}+t_vn$, where $|t_v|$ is bounded by the value $\displaystyle c=\frac{1.5}{m}-\frac{2}{m^2}$.
Since $\displaystyle\sum_{v\in S}z_v$ is the zero vector, it follows the inequality $\displaystyle\Big|\sum_{v\in S}\langle y,z_v\rangle\zeta_v\Big|\le cn\sum_{v\in S}|\langle y,z_v\rangle|$.
Next, by using the observations from the proof of Property~3, we obtain that $\displaystyle\sum_{v\in S}|\langle y,z_v\rangle|\le\|y\|\max_{v\in S}\|z_v\|\,t$, where
$$t=m\cos(\alpha/2)=m\sqrt{\frac{\cos\alpha+1}{2}}=\frac{m}{\sqrt 2}\sqrt{1+\frac{1.01}{m-1}}.$$
Therefore, the multiplier $c\,t$ equals $\displaystyle\frac{1}{\sqrt 2}\Big(1.5-\frac{2}{m}\Big)\sqrt{1+\frac{1.01}{m-1}}$, which is less than $\displaystyle\frac{1.5}{\sqrt 2}<1.061$.
\hfill$\Box$
\end{proof}

\begin{lemma}\label{medmean}
Suppose that $A\ge 100nm$.
Then $\displaystyle\|\mu(X_S)-c(X_S)\|<\frac{4.5n}{mA}$.
\end{lemma}

\begin{proof}
We will compare the values of $cost(X_S,c(X_S))$ and $cost(X_S,\mu(X_S))$.
The first is
$$\sum_{v\in S}\|z_v\|=\sum_{v\in S}\sqrt{A^2+\zeta_v},$$
while the second is
$$\sum_{v\in S}\sqrt{\|z_v\|^2+\|y\|^2-2\langle y,z_v\rangle}=
\sum_{v\in S}\sqrt{A^2+\zeta_v+\delta^2-2\langle y,z_v\rangle}.$$
Note that $\displaystyle\Big|\frac{\zeta_v}{A^2}\Big|<0.001$ and
$$\Big|\frac{\zeta_v+\delta^2-2\langle y,z_v\rangle}{A^2}\Big|<\Big|\frac{1.13n+3.24nm+2\cdot 1.8\sqrt{nm}\sqrt{A^2+1.13n}}{A^2}\Big|<0.02$$
by Property~1, Lemma~\ref{m3}, and the condition for~$A$.
On the other hand, Taylor's theorem (in the Lagrange remainder form) implies the equation
\begin{eqnarray*}
\sqrt{1+\varepsilon}=1+\frac{\varepsilon}{2}-\frac{\varepsilon^2}{8}+\theta\frac{\varepsilon^3}{15.7}\mbox{ for some }\theta\in[0,1]\mbox{ if }|\varepsilon|\le 0.02.
\end{eqnarray*}
Therefore,
$\displaystyle cost(X_S,c(X_S))\le\sum_{v\in S}A\Big(1+\frac{\zeta_v}{2A^2}-\frac{\zeta_v^2}{8A^4}+\frac{\zeta_v^3}{15.7A^6}\Big)$
and
\begin{eqnarray*}
cost(X_S,\mu(X_S))\ge
\sum_{v\in S}A\Big(1+\frac{\zeta_v+\delta^2-2\langle y,z_v\rangle}{2A^2}-\\
\frac{\zeta_v^2+2\delta^2\zeta_v+\delta^4-4\langle y,z_v\rangle(\delta^2+\zeta_v)+4\langle y,z_v\rangle^2}{8A^4}-\frac{8\langle y,z_v\rangle^3}{15.7A^6}\Big).
\end{eqnarray*}
Since $\displaystyle\sum_{v\in S}z_v$ is the zero vector, the sums of the terms $\langle y,z_v\rangle$ are zero.
On the other hand, by Properties~4 and~1, the sum of the terms $\langle y,z_v\rangle\zeta_v$ is greater than $-1.061n\delta\sqrt{A^2+1.13n}$.
At the same time, by Properties~3 and~1, the sum of the terms $-\langle y,z_v\rangle^2$ is greater than $\displaystyle-\frac{m+1.515}{2}\,\delta^2(A^2+1.13n)$.
Therefore, the inequality $cost(X_S,c(X_S))\ge cost(X_S,\mu(X_S))$ implies that
\begin{eqnarray*}
\frac{m\delta^2}{2A^2}-\frac{1.13nm\delta^2}{4A^4}-\frac{m\delta^4}{8A^4}-\frac{1.061n\delta\sqrt{A^2+1.13n}}{2A^4}-\\
\frac{(m+1.515)\,\delta^2(A^2+1.13n)}{4A^4}-
\frac{8m\delta^3(A^2+1.13n)^{3/2}}{15.7A^6}\le\frac{(1.13n)^3m}{15.7A^6}.
\end{eqnarray*}
Recall that $\delta<1.8\sqrt{nm}$ by Lemma~\ref{m3}.
Then, taking into account the obvious inequality $(A^2+1.13n)^\alpha<1.001A^{2\alpha}$ for $\alpha\le 3/2$ and $A\ge 100nm$, we have
\begin{eqnarray*}
\delta\Big(\frac{1}{2}-\frac{1.13n}{4A^2}-\frac{1.8^2nm}{8A^2}-\Big(\frac{1.001}{4}+\frac{1.517}{4m}\Big)-\frac{8\cdot 1.8\sqrt{nm}\cdot 1.001}{15.7A}\Big)<\\
\frac{(1.13n)^3}{15.7\delta A^4}+\frac{1.061\cdot 1.001n}{2mA}.
\end{eqnarray*}
But $m\ge 3$, so
\begin{eqnarray*}
\delta\Big(0.1233-\frac{1.13n}{4A^2}-\frac{1.8^2nm}{8A^2}-\frac{8\cdot 1.8\sqrt{nm}\cdot 1.001}{15.7A}\Big)<
\frac{(1.13n)^3}{15.7\delta A^4}+\frac{1.061\cdot 1.001n}{2mA}.
\end{eqnarray*}
Since $A$ is sufficiently big, it follows that $\displaystyle\delta<\frac{4.5n}{mA}$.
The lemma is proved.
\hfill$\Box$
\end{proof}

\section{The case of $\ell_\infty$ distances}
To prove the statements of Theorem~\ref{th1} for the case of $\ell_\infty$ distances, we construct a similar but much simpler reduction from the Independent Set problem.

Let $G=(V,E)$ be any $n$-vertex undirected graph.
Without loss of generality, we will assume that $G$ has no isolated vertices.
Fix an arbitrary orientation on its edges: for every edge, choose an endpoint of this edge which it is ``outgoing from'' and one which it is ``incoming to''.
Next, map each vertex $v\in V$ to the point $x_v\in\mathbb R^E$ with the following coordinates: $x_v(e)=1$ for every edge $e\in E$ outgoing from $v$ and $x_v(e)=-1$ for ones incoming to $v$; all the other coordinates are zero (see Fig.~\ref{fig:FInfty}).
Note that, since $G$ has no isolated vertices, all the vectors $x_v$, $v\in V$, are distinct.
Define the instance of the $\ell_\infty$-Median problem with outliers corresponding to the graph $G$ as the set $X=\{x_v\mid v\in V\}$.

\begin{figure}
\centering
\includegraphics[scale=1]{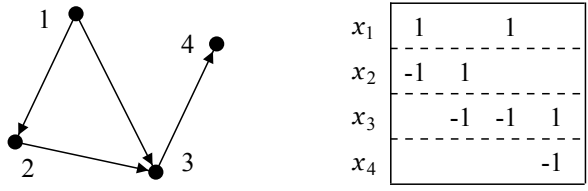}
\caption{Reduction in the case of $\ell_\infty$ distances}
\label{fig:FInfty}
\end{figure}

\begin{Comment}
Selecting an edge orientation provides that the $\ell_\infty$-distance between the vectors $x_v,x_u$ for adjacent different vertices $u,v\in V$ is $2$, while that for non-adjacent ones is~$1$.
It will be used below to justify the proposed reduction.
\end{Comment}

Estimate the value of the $\ell_\infty$-median cost of the set $X_S=\{x_v\mid v\in S\}$ for an arbitrary $m$-element subset of vertices $S\subseteq V$, where $m\ge 2$.
Consider the point $c_S\in\mathbb R^E$ for which $c_S(e)=1/2$ (or $-1/2$) if the edge $e$ is outgoing from $S$ and incoming to $V\setminus S$ (or vice versa),
otherwise, $c_S(e)=0$.
Note that, if the set $S$ is independent, then any edge $e\in E$ which is incident to vertices from $S$ is either outgoing from $S$ or incoming to $S$.
In this case, the $\ell_\infty$-distance from $c_S$ to each point $x_v$, $v\in S$, is exactly $1/2$, so $\displaystyle\sum_{v\in S}\|x_v-c\|_\infty=m/2$.

On the other hand, for any subset $S\subseteq V$ and any point $c\in\mathbb R^E$, we have
\begin{eqnarray*}
\sum_{v\in S}\|x_v-c\|_\infty=\frac{1}{2(m-1)}\sum_{v\in S}\sum_{u\in S\setminus\{v\}}\big(\|x_v-c\|_\infty+\|x_u-c\|_\infty\big)\ge\\
\frac{1}{2(m-1)}\sum_{v\in S}\sum_{u\in S\setminus\{v\}}\|x_v-x_u\|_\infty\ge\frac{m^2-m}{2(m-1)}=m/2.
\end{eqnarray*}
Moreover, by Comment~3, if the set $S$ contains at least two adjacent vertices $a$ and $b$, then
\begin{eqnarray*}
\sum_{v\in S}\|x_v-c\|_\infty=\|x_a-c\|_\infty+\|x_b-c\|_\infty+\sum_{v\in S\setminus\{a,b\}}\|x_v-c\|_\infty\ge\\
\|x_a-x_b\|_\infty+\sum_{v\in S\setminus\{a,b\}}\|x_v-c\|_\infty\ge 2+(m-2)/2=m/2+1.
\end{eqnarray*}
It follows that the $\ell_\infty$-median cost of the set $X_S$ is $m/2$ if the set $S$ is independent and is at least $m/2+1$ otherwise.
Thus, the Independent Set problem is reduced to $\ell_\infty$-Median with outliers and we obtain the following theorem:

\begin{theorem}
The $\ell_\infty$-Median problem with outliers is strongly NP-hard, does not admit approximation schemes FPTAS unless P$=$NP, and is W$[1]$-hard with respect to the parameter~$m$.
\end{theorem}


\end{document}